\documentclass[conference]{IEEEtran}

\usepackage{amsmath,amssymb,amsthm}
\usepackage[pdftex]{color}
\usepackage{enumerate}
\usepackage{lineno}
\usepackage{mathtools}
\usepackage{graphicx}
\usepackage{subfig}
\usepackage{hyperref}

\numberwithin{equation}{section}

\newcommand{\MAT}{\left[ \begin{array}}  
	\newcommand{\mat}{\end{array} \right]}

\newtheorem{theorem}{Theorem}[section]
\newtheorem{proposition}[theorem]{Proposition}

\newtheorem{lemma}[theorem]{Lemma}
\newtheorem{example}[theorem]{Example}

\theoremstyle{definition}

\newtheorem*{definition*}{Definition}
\newtheorem*{proposition*}{Proposition}
\newtheorem*{theorem*}{Theorem}
\newtheorem*{corollary*}{Corollary}
\newtheorem*{example*}{Example}
\newtheorem*{problem*}{Problem}

\theoremstyle{remark}

\newcommand*{\stcomp}[1]{{#1}^{\mathsf{c}}}

\newcommand{\interior}[1]{%
	{\kern0pt#1}^{\mathrm{o}}%
}
\newcommand*{\prob}[1]{\mathbb{P}\left\{ #1 \right\}}

\title{On Extended Concentration Inequalities for Fast JL Embeddings of Infinite Sets}

\author{
	\IEEEauthorblockN{Edem Boahen, March T. Boedihardjo, Rafael Chiclana, Mark Iwen}
	\IEEEauthorblockA{
		Michigan State University \\
		Email: boahened@msu.edu, boedihar@msu.edu, chiclan1@msu.edu, iwenmark@msu.edu
	}
}

\begin{document}
	
	\maketitle

\begin{abstract}
   The Johnson-Lindenstrauss (JL) lemma allows subsets of a high-dimensional space to be embedded into a lower-dimensional space while approximately preserving all pairwise Euclidean distances.  This important result has inspired an extensive literature, with a significant portion dedicated to constructing structured random matrices with fast matrix-vector multiplication algorithms that generate such embeddings for finite point sets.  In this paper, we briefly consider fast JL embedding matrices for {\it infinite} subsets of $\mathbb{R}^d$.  Prior work in this direction such as \cite{oymak2018isometric, mendelson2023column} has focused on constructing fast JL matrices $HD \in \mathbb{R}^{k \times d}$ by multiplying structured matrices with RIP(-like) properties $H \in \mathbb{R}^{k \times d}$ against a random diagonal matrix $D \in \mathbb{R}^{d \times d}$.  However, utilizing RIP(-like) matrices $H$ in this fashion necessarily has the unfortunate side effect that the resulting embedding dimension $k$ must depend on the ambient dimension $d$ no matter how simple the infinite set is that one aims to embed.  Motivated by this, we explore an alternate strategy for removing this $d$-dependence from $k$ herein:  Extending a concentration inequality proven by Ailon and Liberty \cite{Ailon2008fast} in the hope of later utilizing it in a chaining argument to obtain a near-optimal result for infinite sets.  
   Though this strategy ultimately fails to provide the near-optimal embedding dimension we seek, along the way we obtain a stronger-than-sub-exponential extension of the concentration inequality in \cite{Ailon2008fast} which may be of independent interest.
\end{abstract}

\section{Introduction}

The Johnson-Lindenstrauss lemma \cite{johnson1984extensions} states that for $\varepsilon \in (0,1)$ and a finite set $T \subset \mathbb{R}^d$ with $n>1$ elements, there exists a $k\times d$ matrix $\Phi$ with $k=\mathcal{O}(\varepsilon^{-2}\log n)$ such that
\begin{equation} \label{equ:JLprop}
(1-\varepsilon)\|{\bf x}-{\bf y}\|^2_2 \leq \|\Phi {\bf x}-\Phi {\bf y}\|^2_2 \leq (1+\varepsilon)\|{\bf x}-{\bf y}\|_2^2 
\end{equation}
holds $\forall \, {\bf x}, {\bf y} \in T$. A matrix $\Phi$ satisfying (\ref{equ:JLprop}) is called an \textbf{$\varepsilon$-JL embedding} of $T$ into $\mathbb{R}^k$. Moreover, it has been shown that the dimension $k$ of the Euclidean space where $T$ is embedded is optimal for finite sets (see \cite{larsen2017optimality}). This result is a cornerstone in dimensionality reduction and has proved to be an extremely useful tool in many application domains (see, e.g., the relevant discussions in \cite{baraniuk2009random,liberty2011dense, kane2012sparser, ailon2013an, bourgain2015toward, dirksen2016dimensionality,iwen2024on}).


If $T \subset \mathbb{R}^d$ is an infinite set one may bound the embedding dimension $k$ in terms of its {\it Gaussian Width}, $w(T) := \mathbb{E} \sup_{{\bf x} \in T} \, \langle {\bf g},{\bf x} \rangle$, where ${\bf g}$ is a random vector with $d$ independent and identically distributed (i.i.d.) mean $0$ and variance $1$ Gaussian entries (see, e.g,. \cite[Definition 7.5.1]{vershynin2018high-dimensional}). Let ${\rm unit}(T-T) := \{ ({\bf x} - {\bf y})/\|{\bf x} - {\bf y} \|_2 ~|~ {\bf x},{\bf y} \in T, {\bf x} \neq {\bf y} \}$.  For any bounded set \( T \subset \mathbb{R}^d \), standard upper bounds demonstrate that sub-Gaussian random matrices \( \Phi \) are $\varepsilon
$-JL embeddings of $T$ into $\mathbb{R}^k$, where $k=\mathcal{O}(\varepsilon^{-2} w^2({\rm unit}(T-T)))$, with high probability (w.h.p.) (see, e.g., \cite[Theorem 9.1.1 and Exercise 9.1.8]{vershynin2018high-dimensional}). Similarly, \cite[Theorem 9]{iwen2023lower} shows that any
JL embedding of a bounded set $T \subseteq \mathbb{R}^d$ into $\mathbb{R}^k$ must have $k \gtrsim w^2(T)$, which matches the prior upper bound for a large class of sets $T \subset \mathbb{R}^d$ when $\epsilon$ isn't too small. Most importantly for our discussion here, we note that all the bounds on \( k \) mentioned above are entirely independent of \( d \).

If we further demand that a $\Phi \in \mathbb{R}^{k\times d}$ satisfying \eqref{equ:JLprop} $\forall \, {\bf x}, {\bf y} \in T$ also be a structured matrix with an associated fast matrix-vector multiplication algorithm, the situation complicates.  In this setting state-of-the-art results  \cite{oymak2018isometric,mendelson2023column} build on Restricted Isometry Property (RIP) related results implied by, e.g., \cite{krahmer2011new,krahmer2014suprema} to produce structured $\epsilon$-JL embeddings of infinite sets that also have fast matrix-vector multiplication algorithms.  However, their dependence on the RIP has the unfortunate side effect that the bounds they obtain on the embedding dimension \( k \) must always depend (logarithmically) on $d$ no matter how simple the set $T$ is.


Returning to the setting of {\it finite} sets $T$, in \cite{Ailon2008fast} Ailon and Liberty construct $\epsilon$-JL embeddings with fast matrix-vector multiplication algorithms that also have near-optimal embedding dimensions for sets of sufficiently small cardinality.  In particular, they construct a $k \times d$ matrix $A$ for which the mapping ${\bf x} \mapsto A {\bf x}$ can be computed in ${\mathcal O}(d \log k)$-time that also satisfies the following sub-Gaussian concentration inequality: For any ${\bf x} \in \mathbb{R}^d$ with $\| {\bf x}\|_2=1$ and $0<t<1$,
\begin{equation}\label{eq: main Liberty}
	\prob{|\|A {\bf x}\|_2-1|>t}\leq c_1 \exp\{-c_2 k t^2\},
\end{equation}
for some universal constants $c_1$, $c_2>0$. This then allows for optimal dimensionality reduction of finite sets $T$ with $n$ elements by taking $k={\mathcal O}(\epsilon^{-2}\log n)$, where $\epsilon>0$ is the desired distortion of the JL-embedding.  Looking at $\eqref{eq: main Liberty}$ in the context of, e.g., \cite[Chapter 8]{vershynin2018high-dimensional}, one might wonder if $\eqref{eq: main Liberty}$ can be extended to hold for all $t >0$.  If so, a chaining argument could then be employed to extend the fast $\epsilon$-JL embedding results in \cite{Ailon2008fast} to more arbitrary (e.g., infinite) subsets of $\mathbb{R}^d$ with embedding dimensions $k$ that don't depend on $d$.

Unfortunately, the approach in \cite{Ailon2008fast} apparently fails to provide sub-Gaussian concentration for large distortions $t$.  We are, however, at least able to demonstrate an extended concentration inequality herein that is better than sub-exponential.

\begin{theorem}\label{theo: main}
	There is a $k\times d$ random matrix $A$, for which the mapping ${\bf x} \mapsto A {\bf x}$ can be computed in time $\mathcal{O}(d \log k)$, such that for any ${\bf x}\in \mathbb{R}^d$ with $\|{\bf x}\|_2=1$ and $t>0$
	\[\prob{|\|A {\bf x}\|_2-1|>t} \leq c_3\exp\{-c_4 k^{2/3}t^{4/3}\},\]
	for some universal constants $c_3$, $c_4>0$.  
\end{theorem}

The matrix A appearing in the above result is defined as $A = BDHD'$, where $B$ is a $4$-wise independent $k\times d$ matrix, $D$ and $D'$ are independent diagonal matrices with entries that are random variables taking values $\pm 1$ with probability $1/2$, and $H$ is a $d\times d$ Walsh-Hadamard matrix. This is a simplified version of the matrix that Ailon and Liberty constructed in \cite{Ailon2008fast}. While Ailon and Liberty's construction involves iterating the transformation $HD'$ multiple times, Example \ref{example: iterate} demonstrates that such iterations do not lead to improvements in Theorem~\ref{theo: main} for large distortions. We refer to Section \ref{sec 2} for a detailed construction of the matrix $A$.

Using Theorem~\ref{theo: main} one can now quickly prove the following theorem providing $d$-independent bounds on the embedding dimension $k$ of $A$ for infinite sets.

\begin{theorem}\label{thm:ChainingBound}
Let $A$ be the random matrix in Theorem \ref{theo: main}. Let $S\subset\{x\in\mathbb{R}^{d}:\,\|x\|_{2}=1\}$ and $\epsilon,p\in(0,1)$. Then $A$ is an $\epsilon$-JL map of $S$ into $\mathbb{R}^{k}$ with probability at least $1-p$ provided that
\[k\geq\frac{C}{\epsilon^{4}}\left(\left(\ln\frac{1}{p}\right)^{\frac{3}{4}}+\sum_{j=0}^{\infty}\frac{1}{2^{j}}\left(\ln N\left(S,\|\cdot\|_{2},\frac{1}{2^{j}} \right)\right)^{\frac{3}{4}}\right)^{2},\]
where $N\left(S,\|\cdot\|_{2},\frac{1}{2^{j}} \right)$ is the $\frac{1}{2^{j}}$-covering number of $S$, and $C\geq 1$ is a universal constant.
\end{theorem}

Looking at Theorem~\ref{thm:ChainingBound} we can see that it is indeed independent of $d$ as desired.  Unfortunately, it also provides sub-optimal dependence on the covering numbers (in this case) of the set $S$.  It is proven in Section~\ref{sec:proofchainingBound} for completeness.

\section{Preliminaries}\label{sec 2}

In this section, we review the necessary background and results that will be applied throughout the paper.

\subsection{Definitions}

Given ${\bf x} \in \mathbb{R}^d$, we write $\|{\bf x}\|_2$ for the Euclidean norm of ${\bf x}$. Given a $k \times d$ matrix $B$, the operator norm $\|B^T\|_{2 \to 4}$ is defined as the maximum ratio of the $\ell_4$ norm of the matrix-vector product to the $\ell_2$ norm of the vector, formally expressed as:
\[
\|B^T\|_{2 \to 4} = \sup_{\| {\bf x}\|_2 = 1} \|B^T {\bf x}\|_4,
\]
where the $\ell_4$ norm of a vector ${\bf x} \in \mathbb{R}^d$ is defined as
\[
\|{\bf x}\|_4 = \left( \sum_{i=1}^d |x_i|^4 \right)^{1/4}.
\]
A \textit{Rademacher sequence} $\boldsymbol{\xi} \in \mathbb{R}^d$ is a random vector whose coordinates are independent and take the value $1$ or $-1$ with equal probability. A random variable $X$ is said to be \textit{sub-Gaussian} if its \textit{sub-Gaussian norm}, defined as 
\[
\|X\|_{\Psi_2} = \inf\left\{c>0 \colon \mathbb{E} e^{X^2/c^2} \leq 2\right\},
\]
is finite. An equivalent characterization of sub-Gaussianity is given by the tail bound:
\[
\prob{|X| \geq t} \leq 2 e^{-t^2/C_1^2} \quad \forall \, t \geq 0,
\]
for some constant $C_1 > 0$. Another useful characterization is that $X$ is sub-Gaussian if for any $p \geq 1$,
\[
\mathbb{E}|X|^p \leq C_2^p p^{p/2} \quad \text{for some constant } C_2 > 0.
\]
A \textit{Walsh-Hadamard} matrix $H_d$ is a $d \times d$ orthogonal matrix defined recursively: For $d = 1$, $H_1 = [1]$, and for $d = 2^n$, it is defined as
\[
H_d = \frac{1}{\sqrt{2}}\begin{bmatrix}
	H_{d/2} & H_{d/2} \\
	H_{d/2} & -H_{d/2}
\end{bmatrix}.
\]
The entries of a Walsh-Hadamard matrix are given by 
\[
H_d(i,j) = d^{-1/2} (-1)^{\langle i,j \rangle},
\]
where $\langle i,j \rangle$ denotes the dot product of the binary representations of the indices $i$ and $j$. For convenience, we will omit the subscript and denote the Walsh-Hadamard matrix simply as $H$ instead of $H_d$

A matrix $B$ of size $k \times d$ is said to be \textit{4-wise independent} if for any $1 \leq i_1 < i_2 < i_3 < i_4 \leq k$ and any $(b_1,b_2,b_3,b_4) \in \{1,-1\}^4$, the number of columns $B^{(j)}$ for which $(A_{i_1}^{(j)}, A_{i_2}^{(j)}, A_{i_3}^{(j)}, A_{i_4}^{(j)}) = k^{-1/2}(b_1,b_2,b_3,b_4)$ is exactly $d/16$.

\subsection{Supporting Results}

We will utilize several supporting theorems in our analysis. We start with the following classic tool that provides concentration bounds for sums of bounded independent random variables.

\begin{proposition}[Hoeffding's Inequality]\label{prop: Hoeffding}
	Let ${\bf x} \in \mathbb{R}^d$ and let $\boldsymbol{\xi} = (\xi_j)_{j=1}^d$ be a Rademacher sequence. For any $t > 0$,
	\[
	\prob{\left|\sum_{j=1}^d \xi_j x_j\right| > t} \leq 2 \exp\left(-\frac{t^2}{2\|{\bf x}\|^2_2}\right).
	\]
\end{proposition}

We will also need some results from the work of Ailon and Liberty \cite{Ailon2008fast}, particularly the following lemmas:

\begin{lemma}[Corollary 5.1 from \cite{Ailon2008fast}]\label{lem: Liberty corollary 5.1}
	Let $B$ be a $k \times d$ matrix with Euclidean unit length columns, and let $D$ be a random $\{\pm 1\}$ diagonal matrix. Given ${\bf x} \in \mathbb{R}^d$ with $\|{\bf x}\|_2 = 1$, let $Y = \|BD {\bf x}\|_2$. Then, for any $t \geq 0$,
	\[
	\prob{|Y - 1| > t} \leq c_5 \exp\left\{-c_6\frac{t^2}{\|{\bf x}\|_4^2 \|B^T\|^2_{2 \to 4}}\right\},
	\]
	for some universal constants $c_5, c_6 > 0$.
\end{lemma}

\begin{lemma}[Lemma 4.1 from \cite{Ailon2008fast}]\label{lem: lemma 4.1 in Liberty}
	There exists a $4$-wise independent code matrix of size $k \times f_{\operatorname{BCH}}(k)$, where $f_{\operatorname{BCH}}(k) = \Theta(k^2)$.
\end{lemma}

\begin{lemma}[Lemma 5.1 from \cite{Ailon2008fast}]\label{lem: lemma 5.1 in Liberty}
	Assume that $B$ is a $k \times d$ $4$-wise independent matrix. Then, 
	\[
	\|B^T\|_{2 \to 4} \leq  (3d)^{1/4} k^{-1/2}.
	\]
\end{lemma}

\section{Proofs}\label{sec 3}

Consider a $d \times d$ diagonal matrix $D'$, whose diagonal entries are independent variables that take the value $1$ or $-1$ with probability $1/2$.

The norm in $\ell^d_4$ will play an important role in our analysis. Specifically, there exist vectors with $\|{\bf x}\|_2 = 1$ but very small $\|{\bf x}\|_4$, which occurs when the vector is ``flat". Roughly speaking, a vector is considered flat if a substantial portion of its coordinates have similar absolute values. Our first result demonstrates that we can make a unit vector flat w.h.p. by applying $HD'$. This result provides a sharp version of inequality (5.6) in \cite{Ailon2008fast} and serves as a fundamental element of the argument for constructing the desired embedding.

\begin{lemma}\label{lem: lemma 1}
	Let ${\bf x} \in \mathbb{R}^d$ with $\|{\bf x}\|_2=1$. Then, for any $t>0$ we have
	\[ \prob{\|HD'{\bf x}\|_4 \geq t d^{-1/4} } \leq e^{1-c_7 t^4},\]
	where $c_7>0$ is a universal constant.
\end{lemma}

\begin{proof}
	Fix  ${\bf x} \in \mathbb{R}^d$ with $\|{\bf x}\|_2=1$. Write $Z=\sqrt{d}HD'{\bf x}$ and consider random variables $Z_1,\ldots,Z_d$ so that $Z= (Z_1,\ldots,Z_d)^T$. First, since $H$ is orthogonal observe that 
	\begin{equation}\label{eq: step 1}
		\sum_{j=1}^d Z_j^2 = \|Z\|^2_2 = d\|HD'{\bf x}\|_2^2 = d \|D'{\bf x}\|^2_2=d\|{\bf x}\|^2_2=d.
	\end{equation}
	Next, we will show that each $Z_i$ is a sub-Gaussian random variable for $i=1,\ldots,d$. To be more precise, we claim that
	\begin{equation}\label{eq: step 2}
		\prob{|Z_i|>t} \leq 2 e^{-t^2/2}
	\end{equation}
	Fix $i=1,\ldots,d$. Then, $Z_i = d^{1/2} H_{(i)} D' {\bf x}$, where $H_{(i)}$ is the $i^{\rm th}$ row of $H$. Observe that $H_{(i)}D'$ is a vector whose entries are independent with each entry being $d^{-1/2}$ or $-d^{-1/2}$ with probability $1/2$. Consequently, $Z_i$ has the same distribution as $\langle \boldsymbol{\xi},{
    \bf x} \rangle$, where $\boldsymbol{\xi}$ is a Rademacher sequence. Thus, we can apply Proposition \ref{prop: Hoeffding} to obtain the claim.
    
	We now proceed to estimate the norm of $HD'{\bf x}$ in $\ell_4^d$. Observe that for any $p\geq 1$ we have
	\begin{align*} \|Z\|^4_4 &= \sum_{i=1}^d Z_i^4 = \sum_{i=1}^d Z_i^{(2p+2)\frac{1}{p} + 2(1-\frac{1}{p})} \\
		&\leq \left (\sum_{i=1}^d Z_i^{2p+2}\right )^{1/p} \left (\sum_{i=1}^d Z_i^2\right )^{1-\frac{1}{p}}\\
		&=d^{1-\frac{1}{p}} \left (\sum_{i=1}^d Z_i^{2p+2}\right )^{1/p},
	\end{align*}
	where the inequality follows from applying H\"{o}lder's inequality, and the last equality follows from (\ref{eq: step 1}). Taking expectation and using the fact that each $Z_i$ is a sub-Gaussian random variable gives a universal constant $C>0$ for which
	\begin{align*}
		\mathbb{E} (\|Z\|^{4p}_4) &\leq d^{p-1} \sum_{i=1}^d \mathbb{E}|Z_i|^{2p+2} \\
		&\leq d^{p-1} \sum_{i=1}^d C^{2p+2} (2p+2)^{p+1} \\
		&= C^{2p+2} d^p (2p+2)^{p+1}.
	\end{align*}
	Consequently, we deduce that
	\[ \mathbb{E}(\|HD'{\bf x}\|_4^{4p}) \leq C^{2p+2} d^{-p} (2p+2)^{p+1}.\]
	Finally, Markov's inequality gives
	\begin{align*} \prob{\|HD'{\bf x}\|_4 \geq t d^{-1/4}} &\leq \frac{\mathbb{E}(\|HD'{\bf x}\|_4^{4p})}{t^{4p} d^{-p}} \\
		&\leq \frac{C^{2p+2} (2p+2)^{p+1}}{t^{4p}} \\
		&\leq \left (\frac{C_{1}p}{t^4}\right )^{p},
	\end{align*}
    where $C_{1}>0$ is another universal constant. Take $c_{7}=\frac{1}{eC_{1}}$. When $t^{4}\geq eC_{1}$, the result follows by taking $p=\frac{t^4}{eC_{1}}$ since it shows that
	\[\prob{\|HD'{\bf x}\|_4 \geq t d^{-1/4}} \leq \exp \left( -\frac{t^4}{eC_{1}}\right)=e^{-c_7t^4}.\]
    When $t^{4}\leq eC_{1}$, the result is trivial.
        \qedhere
\end{proof}

The following example shows that Lemma \ref{lem: lemma 1} is sharp.

\begin{example}
	Let ${\bf x}=d^{-1/2} (1,1,\ldots,1)^T$. Then, with probability $2^{-d}$ we have $D'{\bf x}={\bf x}$. In such a case, we have that $HD'{\bf x}=H{\bf x}=(1,0,\ldots,0)^T$. Therefore, $\|HD'{\bf x}\|_4=1$. Taking $t=d^{1/4}$, this argument shows that
	\[ \prob{\|HD'{\bf x}\|_4 \geq td^{-1/4}} \geq 2^{-d} \geq e^{-t^4}.\]
\end{example}

As Lemma \ref{lem: lemma 1} shows, applying the transformation \( HD' \) to a vector \( x \in \mathbb{R}^d \) reduces its \(\ell_4\)-norm with high probability. In \cite{Ailon2008fast}, the embedding they consider is an iterative version of ours. Specifically, they apply the transformation \( HD_i \), where \( D_i \) are independent copies of \( D' \), multiple times to further "flatten" the vector. While this approach works well for small distortions, the following example shows that for larger distortions iterating the transformation does not offer any additional improvement.

\begin{example}\label{example: iterate}
    For $i\in \mathbb{N}$, let $D_i$ be indepentend diagonal matrices whose diagonal entries are independent and take the value $1$ or $-1$ with equal probability.
    Let $x \in \mathbb{R}^d$ be a vector whose first $\sqrt{d}$ coordinates are equal to $d^{-1/4}$, and the rest are $0$. Then, $\|x\|_2=1$.  Observe that $D_1x=x$ with probability $2^{-\sqrt{d}}$. In that case, we have $HD_1x=Hx=x$. Repeating this argument $r$ times, we find that with probability at least $2^{-r\sqrt{d}}$,
    \[ HD_rHD_{r-1} \cdots HD_1 x = x\]
    Setting the distortion $t=d^{1/8}$, we obtain
    \[\prob{ \|HD_rHD_{r-1} \cdots HD_1 x\|_4 \geq t d^{-1/4}} \geq 2^{-r\sqrt{d}} \geq e^{-rt^4}.\]
    This probability, up to a constant factor, is the same as the one achieved in Lemma \ref{lem: lemma 1} using just a single iteration.
\end{example}

Let $B$ be a $k\times d$ $4$-wise independent code matrix. Consider $D$ and $D'$ independent diagonal matrices whose diagonal entries are independent and take the value $\pm 1$ with probability $1/2$. Let $H$ be a Walsh-Hadamard matrix of size $d\times d$. Define the matrix $A=BDHD'$.

Now we can prove our main result.

\begin{proof}[Proof of Theorem \ref{theo: main}]
	Let $u>0$ and define the event $E_u=\{\|HD' {\bf x}\|_4 < ud^{-1/4}\}$. Let $\stcomp{E_u}$ be the complement event. Then, conditioning on $E_u$ gives
	\begin{align*} \prob{|\|A{\bf x}\|_{2}-1|>t}
		&\leq \prob{|\|A{\bf x}\|_{2}-1|>t | E_u} + \prob{\stcomp{E_u}}
	\end{align*}
	On the one hand, we can use Lemma \ref{lem: Liberty corollary 5.1} and Lemma \ref{lem: lemma 5.1 in Liberty} to obtain 
	\begin{align*} \prob{|\|A{\bf x}\|_{2}-1|>t | E_u} &\leq c_5 \exp\left \{-c_6\frac{t^2}{u^2 d^{-1/2} \|B^T\|^2_{2\to 4}}\right \} \\
		&\leq c_5 \exp\left \{-c_6\frac{kt^2}{u^2\sqrt{3}} \right \}.
	\end{align*}
	On the other hand, Lemma \ref{lem: lemma 1} gives
	\[ \prob{\stcomp{E_u}} \leq \exp\{1-c_7u^4 \}.\]
	In conclusion, this yields
	\[ \prob{|\|A{\bf x}\|_{2}-1|>t}  \leq c_5 \exp\left  \{-c_6\frac{kt^2}{u^2\sqrt{3}} \right \} + \exp\{1-c_7u^4 \}.\]
	Taking $u=k^{1/6}t^{1/3}$, we conclude that for any $t>0$
	\[ \prob{|\|A{\bf x}\|_{2}-1|>t} \leq c_{3}\exp\{-c_4 k^{2/3}t^{4/3}\},\]
	for some universal constants $c_3$, $c_4>0$.\qedhere
	
\end{proof}

\section{Proof of Theorem~\ref{thm:ChainingBound}}
\label{sec:proofchainingBound}

This section is dedicated to prove \ref{thm:ChainingBound}.

\begin{proof}[Proof of Theorem \ref{thm:ChainingBound}]

Let $j_{\epsilon}=\lceil\log_{2}\frac{8}{\epsilon}\rceil$. For each $j\geq 0$, let $S_{j}\subset S$ be such that $|S_{j}|=N(S,\|\cdot\|_{2},\frac{1}{2^{j}})$ and $S\subset\cup_{{\bf x}_{0}\in S_{j}}\{{\bf x}\in\mathbb{R}^{d}:\,\|{\bf x} -{\bf x}_{0}\|_{2}\leq\frac{1}{2^{j}}\}$. For each ${\bf x}\in S$ and each $j\geq 0$, let $\pi_{j}({\bf x})$ be the closest point in $S_{j}$ to ${\bf x}$. Then $\|\pi_{j}({\bf x})-{\bf x}\|_{2}\leq\frac{1}{2^{j}}$. So
\[\|\pi_{j+1}({\bf x})-\pi_{j}({\bf x})\|_{2}\leq\|\pi_{j+1}({\bf x})-{\bf x}\|_{2}+\|\pi_{j}({\bf x})-{\bf x}\|_{2}\leq\frac{2}{2^{j}}.\]
For each ${\bf x}\in S$, since ${\bf x}=\pi_{j_{\epsilon}}({\bf x})+\sum_{j=j_{\epsilon}}^{\infty}(\pi_{j+1}({\bf x})-\pi_{j}({\bf x}))$, we have
\begin{align*}
&\left|\|A{\bf x}\|_{2}-1\right|\leq|\|A\pi_{j_{\epsilon}}({\bf x})\|_{2}-1|+\|A({\bf x}-\pi_{j_{\epsilon}}({\bf x}))\|_{2}\\&\leq
|\|A\pi_{j_{\epsilon}}({\bf x})\|_{2}-1|+\sum_{j=j_{\epsilon}}^{\infty}\|A(\pi_{j+1}({\bf x})-\pi_{j}({\bf x}))\|_{2}.
\end{align*}
Since $\pi_{j_{\epsilon}}({\bf x})\in S_{j_{\epsilon}}$ and $\pi_{j+1}({\bf x})-\pi_{j}({\bf x})\in S_{j+1}-S_{j}$ for all ${\bf x}\in S$, by Theorem \ref{theo: main}, with probability at least $1-\frac{p}{2}$, we have
\begin{align*}
&\sup_{{\bf x}\in S}|\|A\pi_{j_{\epsilon}}({\bf x})\|_{2}-1|\\\leq&
\frac{1}{\sqrt{k}}\left(\frac{1}{c_{4}}\ln\left(\frac{2c_{3}|S_{j_{\epsilon}}|}{p}\right)\right)^{\frac{3}{4}}\\\leq&
\frac{1}{\sqrt{k}}\left((\ln|S_{j_{\epsilon}}|)^{\frac{3}{4}}+\left(\ln\frac{1}{p}\right)^{\frac{3}{4}}+\left(\frac{\ln 2c_{3}}{c_{4}}\right)^{\frac{3}{4}}\right),
\end{align*}
and with probability at least $1-\frac{p}{2^{j+2}}$, we have
\begin{align*}
&\sup_{{\bf x}\in S}\left\|A\frac{\pi_{j+1}({\bf x})-\pi_{j}({\bf x})}{\|\pi_{j+1}({\bf x})-\pi_{j}({\bf x})\|_{2}}\right\|_{2}\\\leq&
1+\frac{[\frac{1}{c_{4}}\ln(\frac{2^{j+2}c_{3}}{p}|S_{j+1}||S_{j}|)\,]^{\frac{3}{4}}}{\sqrt{k}}\\\leq&
1+\frac{1}{\sqrt{k}}\bigg((\ln|S_{j+1}|)^{\frac{3}{4}}+(\ln|S_{j}|)^{\frac{3}{4}}+\left(\ln\frac{1}{p}\right)^{\frac{3}{4}}\\&
+\left((j+2)\ln 2\right)^{\frac{3}{4}}+\left(\frac{\ln c_{3}}{c_{4}}\right)^{\frac{3}{4}}\bigg).
\end{align*}
Therefore, with probability at least $1-p$, we have
\begin{align*}
&\sup_{{\bf x}\in S}|\|A{\bf x}\|_{2}-1|\\\leq&
\frac{1}{\sqrt{k}}\left((\ln|S_{j_{\epsilon}}|)^{\frac{3}{4}}+3\left(\ln\frac{1}{p}\right)^{\frac{3}{4}}+C+\sum_{j\geq j_{\epsilon}}\frac{6}{2^{j}}(\ln|S_{j}|)^{\frac{3}{4}}\right)\\&\hspace{.3in}+\frac{4}{2^{j_{\epsilon}}}\\\leq&
\frac{C}{\sqrt{k}}\left(\left(\ln\frac{1}{p}\right)^{\frac{3}{4}}+\frac{1}{\epsilon}\sum_{j\geq j_{\epsilon}}\frac{1}{2^{j}}(\ln|S_{j}|)^{\frac{3}{4}}\right)+\frac{\epsilon}{2}.\qedhere
\end{align*}

\end{proof}

\end{document}